\documentclass[
  twocolumn,
  aps,
  prl,
  showpacs,
  superscriptaddress,
  preprintnumbers,
  letterpaper
]{revtex4}

\usepackage{amsthm}
\usepackage{amsmath}
\usepackage{amssymb}
\usepackage{bbold}
\usepackage{titlesec}

\newtheorem{theorem}{Theorem}
\newtheorem{lemma}{Lemma}
\newtheorem{corollary}{Corollary}

\theoremstyle{definition}

\titlespacing{\section}{0mm}{1.5ex plus .1ex minus .2ex}{5pt}
\titleformat{\section}[runin]{\itshape}{\thesection}{.5em}{}[.]

\newcommand*{\cW}{\mathcal{W}}
\newcommand*{\cX}{\mathcal{X}}
\newcommand*{\cY}{\mathcal{Y}}
\newcommand*{\cZ}{\mathcal{Z}}

\newcommand{\ket}[1]{|#1\rangle}
\newcommand{\bra}[1]{\langle #1 |}
\newcommand{\unidim}[2]{\ket{#1}\bra{#2}}

\newcommand{\proj}[1]{\ket{#1}\bra{#1}}

\newcommand{\ot}[0]{\otimes}

\newcommand{\beq}{\begin{equation}}
\newcommand{\eeq}{\end{equation}}
\newcommand{\best}{\begin{equation*}}
\newcommand{\eest}{\end{equation*}}

\newcommand{\sep}{{\textrm{sep}}}

\newcommand{\Sep}{{\rm Sep}}

\newcommand{\tp}{{\rm TP}}
\newcommand{\ta}{{\rm TA}}

\DeclareMathOperator{\Tr}{Tr}
\newcommand{\op}[1]{\operatorname{#1}}
\newcommand{\setft}[1]{\mathrm{#1}}
\newcommand{\lin}[1]{\setft{L}\left(#1\right)}
\newcommand{\density}[1]{\setft{D}\left(#1\right)}
\newcommand{\trans}[1]{\setft{T}\left(#1\right)}

\def\I{\mathbb{1}}
\def\complex{\mathbb{C}}

\newenvironment{mylist}[1]{\begin{list}{}{
	\setlength{\leftmargin}{#1}
	\setlength{\rightmargin}{0mm}
	\setlength{\labelsep}{2mm}
	\setlength{\labelwidth}{8mm}
	\setlength{\itemsep}{0mm}}}
	{\end{list}}

\begin{document}

\title{All entangled states are useful for channel discrimination}

\author{Marco Piani}
\affiliation{Institute for Quantum Computing \& Department of Physics
  and Astronomy, University of Waterloo, 200 University Avenue West, 
  Waterloo, Ontario N2L 3G1, Canada}

\author{John Watrous}
\affiliation{Institute for Quantum Computing \& School of Computer
  Science,
  University of Waterloo, 200 University Avenue West, 
  Waterloo, Ontario N2L 3G1, Canada}


\begin{abstract}
  We prove that every entangled state is useful as a resource for the
  problem of minimum-error channel discrimination.
  More specifically, given a single copy of an arbitrary bipartite
  entangled state, it holds that there is an instance of a quantum
  channel discrimination task for which this state allows for a
  correct discrimination with strictly higher probability than 
  every separable state.
\end{abstract}

\maketitle

Despite its sometimes counter-intuitive properties, entanglement has
firmly been established as a fundamental resource at the core of
quantum information theory.
Universal quantum computation is generally believed to be impossible
in its absence~\cite{JoszaL03}, and it plays a principal role in
quantum teleportation~\cite{BennettBCJPW93}, superdense
coding~\cite{BennettW92}, and the one-way model of quantum
computation~\cite{RaussendorfB01}.
The classification of entanglement
%into different types, depending on
with respect to its usefulness and properties as a resource is a major focus in the
theory of quantum information.
For example, {\it distillable entanglement}~\cite{BennettBPSSW96} may
be processed by means of local operations and classical communication
into a nearly pure form that is suitable for high fidelity quantum
teleportation, while {\it bound entanglement} cannot~\cite{HorodeckiHH98}.
Other classifications of entangled states, such as those that allow
or do not allow superdense coding~\cite{BrussDLMSS04,HododeckiP07},
and those from which private shared-randomness can be extracted
\cite{HorodeckiHHO05}, have also been studied.

Although entanglement is known to be useful in several quantum
information-theoretic settings, there are very few known results that
establish the usefulness of \emph{every} entangled state, irrespective
of the ``quality'' of its entanglement and of the dimensionality of
its underlying systems.
The only prior examples that we are aware of involve a type of
activation mechanism, where the usefulness of a given entangled state
is based on
%the joint properties of a composite system when it is
its pairing with another entangled state.
%of a special type.
For example, in~\cite{Masanes05} it was proved that for any entangled
state, there exist another entangled state such that the fidelity of
conclusive teleportation~\cite{HorodeckiHH99} of the latter is enhanced
by the presence of the former. 
A different property holding for all entangled states that has a
similar character was proved in \cite{MasanesLD08}.

In this Letter we demonstrate a new way in which every entangled state
is useful as a resource: for the task of {\it channel discrimination}.
In this task, two known discrete physical processes (or channels)
are fixed, and access to one of them is made available---but it is not
known which one it is, and only a single application of the channel is
possible.
The goal is to determine, with minimal probability of error, which of
the two channels was given, assuming for simplicity that the two
channels were equally likely.
The most general approach to solving an instance of this problem is to
prepare a (possibly entangled) bipartite \emph{probe/ancilla} quantum state, to apply the
given channel to one part of this state---the probe---and finally to measure the
resulting bipartite state by a POVM with two outcomes that correspond to
predictions of which channel was given.

It is well-known that probe-ancilla entanglement is sometimes useful for channel
discrimination.
This phenomenon seems to have been identified first by Kitaev
\cite{Kitaev97}, who introduced the {\it diamond norm} on
super-operators to deal with precisely this phenomenon in the context
of quantum error correction and fault-tolerance
\footnote{
	The diamond norm turns out to be essentially equivalent to a norm
	known as the {\it completely bounded norm}, which is an object of
	study in the theory of operator algebras \cite{Paulsen02}.
}.
Subsequent work 
\cite{
  ChildsPR00,
  D'ArianoPP01,
  Acin01,
  GiovannettiLM04,
  GilchristLN05,
  RosgenW05,
  Sacchi05,
  Sacchi05b,
  Lloyd08,
  Rosgen08,
  Watrous08}
by several researchers further illuminated the usefulness of
entanglement in the problem of channel discrimination and related
tasks.
In these works, the focus has mainly been on identifying classes of
channel pairs for which some optimally chosen entangled state either
does or does not give an advantage over every possible separable (or
nonentangled) state.

In this Letter we reverse this question and suppose that some
{\it arbitrary} entangled state is given, and ask whether the
entanglement in this state is useful for channel discrimination.
We prove that every bipartite entangled state indeed does provide an
advantage for this task: \emph{there necessarily exists an instance
of a channel discrimination problem for which the entangled state
allows for a correct discrimination with strictly higher probability
than every possible separable state}.
This holds even for a single copy of the entangled state, regardless
of its dimensionality or the quality or type of its entanglement
(including, for instance, bound entangled states), and does not
require the presence of an auxiliary state.
%that serves to activate it.
This fact is proved below after brief discussions of notation,
terminology, and background information on the problem of channel
discrimination.

%--------------------------------------------------------------------------%
\section{Notation and terminology}
%--------------------------------------------------------------------------%

For a given (finite dimensional) Hilbert space $\cX$, the set of
linear operators taking the form $A:\cX\rightarrow\cX$ is denoted by
$\lin{\cX}$. We will denote by $\I_\cZ$ the identity operator on $\cZ$ and by $\I_{\lin{\cZ}}$
the identity super-operator on $\lin{\cZ}$. 
An operator $\rho\in\lin{\cX}$ is a {\it density operator}, and
represents a {\it state}, if it is positive semidefinite
($\rho\geq0$) and has unit trace ($\Tr(\rho)=1$).
The set of such density operators is denoted $\density{\cX}$.
A state $\sigma\in\density{\cX\otimes\cZ}$ of a bipartite system
is said to be {\it separable} if it takes the form
%\begin{equation} \label{eq:separable}
$\sigma=\sigma^{\rm sep} = \sum_i{p_i}\, \sigma^i_\cX\otimes\sigma^i_\cZ$
%\end{equation}
for density operators $\{\sigma^i_\cX\}$ and $\{\sigma^i_\cZ\}$ on the
Hilbert spaces $\cX$ and $\cZ$, respectively, and $\{p_i\}$ a probability distribution, and otherwise is
{\it entangled}.
The set of all separable states %of the above form \eqref{eq:separable}
is denoted $\Sep(\cX:\cZ)$.
The {\it trace norm} of an operator $A$ is defined as
$\|A\|_{\rm tr}\equiv\Tr\sqrt{A^\dagger A}$
\footnote{
  When $A=\sum_i a_i \proj{i}$ is Hermitian, the trace norm coincides
  with the sum of the absolute values of the eigenvalues
  $\|A\|_{\rm tr}=\sum_i|a_i|$.
}.
The {\it trace distance} between two states $\rho_0$ and $\rho_1$ is
$\|\rho_0-\rho_1\|_{\rm tr}$.

Channels are particular elements of the set of linear super-operators
$\trans{\cX,\cY}\equiv\{\Phi|\Phi:\lin{\cX} \rightarrow \lin{\cY}\}$ that map
operators on a Hilbert space $\cX$ into operators on a (possibly
different) Hilbert space $\cY$. 
A super-operator $\Phi \in \trans{\cX,\cY}$ is said to be:
\begin{mylist}{5mm}
\item[$\bullet$]
  \emph{Hermiticity-preserving}
%if $\Phi[X]$ is Hermitian for every
%  Hermitian operator $X$; or equivalently
  if $\Phi[X]^\dagger=\Phi[X^{\dagger}]$, $\forall X\in\lin{\cX}$;
\item[$\bullet$]
  \emph{trace-preserving} if $\Tr(\Phi[X])=\Tr(X)$, $\forall X\in\lin{\cX}$;
\item[$\bullet$]
  \emph{trace-annihilating} if $\Tr(\Phi[X])=0$, $\forall X\in\lin{\cX}$;
\item[$\bullet$]
  \emph{positive} if $\Phi[X]\geq 0$ for every positive semidefinite
  operator $\lin{\cX}\ni X\geq 0$;
\item[$\bullet$]
  \emph{completely positive} if $\Phi\otimes\I_{\lin{\mathbb{C}^n}}$
  is positive for all $n$;
\item[$\bullet$]
  a \emph{channel} if it is both completely positive and
  trace-preserving;
\item[$\bullet$] an \emph{entanglement-breaking channel} if it is a
  channel that destroys all entanglement:
  $\left(\Phi\otimes\I_{\lin{\cZ}}\right)[\rho_{\cX\cZ}]\in\Sep(\cY:\cZ)$
  for all states $\rho_{\cX\cZ}$.
\end{mylist}
A channel describes any physical process which preserves probability,
i.e., that happens with certainty.

The {\it Choi-Jamio{\l}kowski representation}
\cite{Jamiolkowski72,Choi75} of a super-operator
$\Phi\in\trans{\cX,\cY}$ is given by
\[
J(\Phi) = \sum_{1\leq i,j\leq d_{\cX}}
\Phi[\unidim{i}{j}] \otimes \unidim{i}{j}
\,\in\,\lin{\cY\otimes\cX},
\] 
where $d_{\cX}$ and $\{\ket{1},\ldots,\ket{d_{\cX}}\}$ are the dimension
and a fixed orthonormal basis of $\cX$, respectively. 
The mapping $J:\trans{\cX,\cY}\rightarrow\lin{\cY\otimes\cX}$ is a
linear bijection, which implies that for every operator
$A\in\lin{\cY\otimes\cX}$ there exists a unique super-operator
$\Phi\in\trans{\cX,\cY}$ such that $J(\Phi) = A$.
It holds that a super-operator $\Phi\in\trans{\cX,\cY}$ is: 
\begin{mylist}{5mm}
\item[$\bullet$] \emph{Hermiticity-preserving} if and only if
  $J(\Phi)^\dagger=J(\Phi)$ \cite{dePillis67};
\item[$\bullet$] \emph{trace-preserving} if and only if
  $\Tr_\cY(J(\Phi))=\I_{\cX}$; 
\item[$\bullet$] \emph{trace-annihilating} if and only if
  $\Tr_\cY(J(\Phi))=0$;
\item[$\bullet$] \emph{completely positive} if and only if
  $J(\Phi)\geq 0$ \cite{Jamiolkowski72,Choi75};
\item[$\bullet$] an \emph{entanglement-breaking channel} if and only if it
  is a channel and $J(\Phi)/d_{\cX} \in \Sep(\cY\otimes\cX)$
  \cite{HorodeckiSR03}.
\end{mylist}
%(The second and third properties above are straightforward.)

%--------------------------------------------------------------------------%
\section{State and channel discrimination}
%--------------------------------------------------------------------------%

The task of channel discrimination is naturally related to the
well-studied task of discriminating states \cite{Helstrom69}.
Suppose we are given one of two known states $\rho_0,\rho_1\in D(\cX)$, each with equal \emph{a priori} probability,
and our goal is to guess which one it is with minimal error
probability.
A guessing procedure for this task may be described by a two-outcome
POVM $\{M_0,M_1\}\subset \lin{\cX}$, $M_0,M_1\geq 0$, $M_0+M_1=\I_{\cX}$.
The error probability for such a measurement can be expressed as
$p_E=1/2(1-1/2\Tr[(M_0-M_1)(\rho_0-\rho_1)])$.
%as
%quantum states are not perfectly distinguishable in general, this
%probability of error
It may be nonzero for every possible measurement, but by optimizing the measurement one reaches the
minimum error probability $p^{\rm min}_E=1/2(1-1/2\|\rho_0-\rho_1\|_{\rm tr})$
\footnote{
  This is achieved by choosing $M_0$ and $M_1$ to be the projectors on
  the positive and negative subspaces of $\rho_0-\rho_1$.
}.

Now, suppose we want to discriminate two channels 
$\Phi_0,\Phi_1\in\trans{\cX,\cY}$ with minimal error probability, as
discussed above.
By ``probing'' whichever channel was given with a state
$\rho\in\density{\cX}$, we transform the problem into one of
discriminating between the states $\Phi_0[\rho]$ and $\Phi_1[\rho]$.
Thus, the relevant quantity becomes
$\|\Phi_0[\rho]-\Phi_1[\rho]\|_{\rm tr}$, and the minimal error will
be achieved by choosing an optimal input state that minimizes
this quantity.
In this way we are led to consider the trace distance
\footnote{This norm is different from the super-operator
  norm that is induced by the trace norm, which is sometimes also
  denoted $\|\cdot\|_{\mathrm{tr}}$.}
of two channels
$\|\Phi_0-\Phi_1\|_{\rm tr}\equiv\max_{\rho}
\|\Phi_0[\rho]-\Phi_1[\rho]\|_{\rm tr}$.
By the convexity of the trace norm, this maximum will be achieved for
some pure input state.

As mentioned previously, however, the reduction from channel to state
discrimination just described may not always be optimal, for it does
not exploit the possibility of feeding the channel with a subsystem
of a larger correlated system, and then measuring the resulting output joint system.
More precisely, we may consider an input state 
$\rho\in\density{\cX\otimes\cZ}$, with $\cZ$ the Hilbert
space of an arbitrary ancillary system, and compare the output states 
$\left(\Phi_i\otimes\I_{\lin{\cZ}}\right)[\rho]$, for $i = 0,1$.
Thus, the ultimate quantity relevant in minimal-error channel
discrimination is actually the diamond norm~\footnote{The supremum is always achieved for $n\leq \op{dim}(\cX)$.}:
\[
\|\Phi_0-\Phi_1\|_{\diamond}\equiv \sup_{n\geq 1}\|\Phi_0\otimes
\I_{\lin{\mathbb{C}^n}}-\Phi_1\otimes\I_{\lin{\mathbb{C}^n}}\|_{\rm tr}.
\]
By definition, it holds that
$\|\Phi_0-\Phi_1\|_{\diamond}\geq\|\Phi_0-\Phi_1\|_{\rm tr}$, and if
it is the case that 
$\|\Phi_0-\Phi_1\|_{\diamond}>\|\Phi_0-\Phi_1\|_{\rm tr}$, 
then it is necessarily because of entanglement.
%This is due to the fact that
Indeed, the correlations of separable states
never help in the discrimination of channels, as for every separable state $\sigma^{\rm sep}\in \Sep(\cX\otimes\cZ)$ we have:
\begin{multline*}
\left\|
\left(\Phi_0\otimes\I_{\lin{\cZ}}\right)[\sigma^{\rm sep}]-\left(\Phi_1\otimes\I_{\lin{\cZ}}\right)[\sigma^{\rm sep}]
\right\|_{\mathrm{tr}}\\
\begin{aligned}
& \leq\;\sum_i p_i 
\|(\Phi_0-\Phi_1)[\sigma^i_\cX]\otimes\sigma^i_\cZ\|_{\rm tr}\\
&= \;\sum_i p_i \|\Phi_0[\sigma^i_\cX]-\Phi_1[\sigma^i_\cX]\|_{\rm tr}
\;\leq\; \|\Phi_0-\Phi_1\|_{\rm tr}.
\end{aligned}
\end{multline*}
% \begin{multline*}
% \left\|
% \left(\Phi_0\otimes\I_{\lin{\cZ}}\right)[\sigma^{\rm sep}]
% - \left(\Phi_1\otimes\I_{\lin{\cZ}}\right)[\sigma^{\rm sep}]
% \right\|_{\mathrm{tr}}\\
% \leq\;\sum_i p_i 
% \|\Phi_0[\sigma^i_\cX]\otimes\sigma^i_\cZ-
% \Phi_1[\sigma^i_\cX]\otimes\sigma^i_\cZ\|_{\rm tr}\\
% = \;\sum_i p_i \|\Phi_0[\sigma^i_\cX]-\Phi_1[\sigma^i_\cX]\|_{\rm tr}
% \;\leq\; \|\Phi_0-\Phi_1\|_{\rm tr}
% \end{multline*}

%--------------------------------------------------------------------------%
\section{Proof of the main result}
%--------------------------------------------------------------------------%

%We will now present a proof of our main result.
To establish our main result, we will connect the characterization of
entanglement in terms of positive linear maps with its usefulness for
channel discrimination.

% We begin with the following lemma, which is a simplification of
% Lemma~1 in \cite{HorodeckiHH06}.
% This lemma states that the well-known characterization of entanglement
% by positive maps proved in \cite{HorodeckiHH96} continues to hold if
% an extra constraint is placed on the positive maps: that they preserve
% trace.
% In addition to a significantly simpler proof, we obtain a better bound
% on the dimension of the output space $\cY$ in comparison to
% \cite{HorodeckiHH06}.
We begin with a lemma that can be considered an improvement of
Lemma~1 in \cite{HorodeckiHH06}: the well-known characterization of entanglement
by positive maps proved in \cite{HorodeckiHH96} continues to hold if
the extra constraint of trace-preservation is placed on the positive maps.
The improvement of the following lemma lies in a significantly simpler proof and in a better bound on the
output dimension of the positive maps.

\begin{lemma}
  \label{lemma:TPdetection}
  A state $\rho\in\density{\cX\otimes\cZ}$ is entangled if and only if
  there exists a positive, trace-preserving super-operator
  $\Phi_{\rm TP} \in \trans{\cX,\cY}$ such that
  \begin{equation} \label{eq:entanglement-witness}
    \left(\Phi_\tp\otimes\I_{\lin{\cZ}}\right)[\rho] \ngeq 0.
  \end{equation}
  It suffices to take $\dim\cY\leq\dim\cZ+1$.
\end{lemma}

\begin{proof}
In \cite{HorodeckiHH96} it was proved that a state 
$\rho\in \density{\cX\otimes\cZ}$ is entangled if and only if there
exists a positive super-operator $\Omega \in \trans{\cX,\cZ}$ such
that $(\Omega\otimes\I_{\lin{\cZ}})[\rho]\ngeq 0$.
The main issue that must be addressed is that the super-operator
$\Omega$ may not, in general, be trace-preserving.

Let us define
$\lambda(\Omega)\equiv\max_{\rho}\Tr(\Omega[\rho])$,
%\footnote{
%  One immediately checks that
%  $\lambda_{\Tr}(\Omega)=\|\Omega^\dagger(\I_{\lin{\cY}})\|$,
%  where $\|A\|$ is the usual operator norm of $A$, 
%  and $\Omega^\dagger\in\trans{\cZ,\cX}$ is the adjoint super-operator
%  to $\Omega$.},
where the maximum is over all density operators
$\rho\in\density{\cX}$, and consider the normalized map
$\hat{\Omega}\equiv\Omega/\lambda(\Omega)$. 
By construction, this super-operator satisfies 
$\Tr(X)\geq \Tr(\hat{\Omega}[X])$ for all $X\geq 0$, and so the map
$\Phi_{\rm TP}\in \trans{\cX,\cZ\oplus\complex}$ defined as
$\Phi_{\rm TP} [X] \equiv 
\hat{\Omega}[X] + (\Tr(X)-\Tr(\hat{\Omega}[X]))\,\proj{0}$,
where $\ket{0}$ is a normalized vector orthogonal to $\cZ$, is also
positive and satisfies 
$(\Phi_{\rm TP}\otimes\I_{\lin{\cZ}})[\rho]\ngeq 0$. 
By taking $\cY = \cZ\oplus\complex$ and noticing that $\Phi_{\rm TP}$ is trace-preserving, 
the proof is complete.
\end{proof}

%\begin{remark}
%   For the sake of the proof of the main theorem below,
It is helpful
  to note at this point that any positive and trace-preserving super-operator $\Phi_\tp\in \trans{\cX,\cY}$ allows one to define, for every
  state $\rho\in\density{\cX\otimes\cZ}$, a generalized \emph{negativity}~\cite{ZyczkowskiHSL98,VidalW02} parameter as~\footnote{Note that the party on which the super-operator is applied is in general relevant.}
\[
N_{\Phi_\tp}(\rho)\equiv\frac{\|\left(\Phi_\tp\otimes\I_{\lin{\cZ}}\right)[\rho]\|_{\rm tr}-1}{2}=\sum_{i:r_i<0}|r_i|,
\]
where $\{r_i\}$ is the set of eigenvalues of $\left(\Phi_\tp\otimes\I_{\lin{\cZ}}\right)[\rho]$. Of course,
    $\left(\Phi_\tp\otimes\I_{\lin{\cZ}}\right)[\sigma^{\rm sep}] \geq 0$
    and $N_{\Phi_\tp}(\sigma^{\sep})=0$, for every
    separable state $\sigma^{\rm sep}\in\Sep(\cX:\cZ)$, while $N_{\Phi_\tp}(\rho)>0$ if $\rho$ is entangled and detected as in \eqref{eq:entanglement-witness}.
%\end{remark}

Next we will prove a lemma that relates a Hermiticity-preserving,
trace-annihilating super-operator---an apparently abstract object---to the existence of two channels.

\begin{lemma} \label{lem:TA}
  Let $\Phi_\ta \in \trans{\cX,\cY}$ be a Hermiticity preserving,
  trace-annihilating super-operator. 
  Then there exist channels $\Psi_0,\Psi_1\in\trans{\cX,\cY}$ and a
  scalar $c_{\Phi_\ta}>0$ such that $c_{\Phi_\ta}\Phi_\ta=\Psi_0-\Psi_1$. 
\end{lemma}

\begin{proof}
  Given that $\Phi_\ta$ is Hermiticity-preserving and trace-annihilating,
  it holds that its Choi-Jamio{\l}kowski representation $J(\Phi_\ta)$ is
  Hermitian and satisfies $\Tr_\cY J(\Phi_\ta) = 0$.
  Let $J(\Phi_\ta) = P_0 - P_1$ be a Jordan decomposition of $J(\Phi_\ta)$
  (meaning that $P_0,P_1\geq 0$ and $\Tr(P_0 P_1)=0$), and note that
  $\Tr_\cY P_0 = \Tr_\cY P_1 =: Q \geq 0$.
  Take $c_{\Phi_\ta} = 1/\|Q\|$, so that $c_{\Phi_\ta} Q \leq \I_{\cX}$.
  Next, consider any positive operator $\xi\in L(\cY\otimes\cX)$ such that
  $\Tr_\cY\xi_{\cY\cX} = \I_{\lin{\cX}} - c_{\Phi_\ta} Q$
  \footnote{
    One possible canonical choice, uniquely defined up to an isometry on
    $\cY$, is a (non-normalized) purification
    $\xi=\proj{\xi}$, $\ket{\xi}\in \cY\otimes\cX$, of $\I_{\lin{\cX}} - c_{\Phi_\ta} Q$.
  }, and let $\Psi_0,\Psi_1\in\trans{\cX,\cY}$ be the unique
  super-operators for which $J(\Psi_i)=c_{\Phi_\ta} P_i + \xi$ for $i=0,1$. 
  We have $J(\Psi_i) \geq 0$ and 
  $\Tr_\cY(J(\Psi_i)) = c_\Phi Q+\I_{\cX}-c_{\Psi_\ta} Q = \I_{\cX}$,
  therefore $\Psi_0,\Psi_1$ are channels.
  Moreover, $J(\Psi_0)-J(\Psi_1) = c_{\Phi_\ta}(P_0-P_1)=c_{\Phi_\ta} J(\Phi_\ta)$, 
  therefore $\Psi_0 - \Psi_1 = c_{\Phi_\ta} \Phi_\ta$.
\end{proof}

We are now ready for the proof of the main theorem, which will rely on the careful definition of a trace-annihilating map---starting from a trace-preserving map as in Lemma~\ref{lemma:TPdetection}---and on the application of Lemma~\ref{lem:TA}.

\begin{theorem} \label{thm:main}
  A state $\rho\in\density{\cX\ot\cZ}$ is entangled if and only if
  there exist channels $\Psi_0,\Psi_1\in\trans{\cX,\cY}$ such that
  \[
  \left\|\left(\Psi_0\otimes\I_{\lin{\cZ}}\right)[\rho] -
  \left(\Psi_1\otimes\I_{\lin{\cZ}}\right)[\rho]\right\|_{\rm tr}
  > \|\Psi_0-\Psi_1\|_{\rm tr}.
  \]
  It suffices to take $\dim\cY\leq\dim\cZ+2$. 
\end{theorem}

\begin{proof}
We have already argued that if $\rho$ allows, for some choice of
channels $\Psi_0,\Psi_1$, a discrimination better than the one
corresponding to $\|\Psi_0-\Psi_1\|_{\rm tr}$, then $\rho$ must be
entangled. 
On the other hand, if $\rho$ is entangled, then by
Lemma~\ref{lemma:TPdetection} there exists a positive,
trace-preserving super-operator $\Phi_\tp \in \trans{\cX,\cW}$ such that
%$\|(\Phi_\tp\otimes\I_{\lin{\cZ}})[\rho]\|_{\rm tr} =
%1+N_{\Phi_\tp}(\rho)>1$. 
$N_{\Phi_\tp}(\rho)>0$. 
Let us define a new map
$\Phi_{\rm TA}\in \trans{\cX,\cW\oplus\complex}$ as
$\Phi_{\rm TA}[X] \equiv \Phi_\tp[X] - \Tr(X)\proj{0}$,
where $\ket{0}$ is a normalized vector orthogonal to $\cW$. 
By construction, $\Phi_{\rm TA}$ is Hermiticity-preserving and
trace-annihilating. 
By Lemma \ref{lem:TA}, there exists a scalar $c_{\Phi_{\rm TA}}$ such
that $c_{\Phi_{\rm TA}}\Phi_{\rm TA}=\Psi_0-\Psi_1$ for two channels
$\Psi_0,\Psi_1\in \trans{\cX,\cW\oplus\complex}$.

Now, for a generic state $\tau\in \density{\cX\ot\cZ}$, one finds 
\begin{multline*}
\|
\left((\Psi_0-\Psi_1)\otimes\I_{\lin{\cZ}}\right)[\tau]\|_{\rm tr}
\hspace{-5cm}\\
\begin{aligned}
&=c_{\Phi_{\rm TA}}\|(\Phi_{\rm TA}\otimes\I_{\lin{\cZ}})[\tau]\|_{\rm tr}\\
&=c_{\Phi_{\rm TA}}\|(\Phi_\tp\otimes\I_{\lin{\cZ}})[\tau] -
\proj{0}\otimes \Tr_\cX(\tau)\|_{\rm tr}\\
&=c_{\Phi_{\rm TA}}(1+\|(\Phi_\tp\otimes\I_{\lin{\cZ}})[\tau]\|_{\rm tr})\\
&=2c_{\Phi_{\rm TA}}(1+N_{\Phi_\tp}(\tau)).
\end{aligned}
\end{multline*} 
For every separable state $\sigma^{\rm sep}\in\Sep(\cX:\cZ)$ we
obtain $\|((\Psi_0-\Psi_1)\otimes\I_{\lin{\cZ}})[\sigma^{\rm sep}]\|_{\rm tr}
= 2c_{\Phi_{\rm TA}}$,
and therefore $\|\Psi_0-\Psi_1\|_{\rm tr}=2c_{\Phi_{\rm TA}}$. 
Thus,
\begin{multline*}
\|\left((\Psi_0-\Psi_1)\otimes\I_{\lin{\cZ}}\right)[\rho]\|_{\rm tr}
- \|\Psi_0-\Psi_1\|_{\rm tr}\\
= 2c_{\Phi_{\rm TA}}N_{\Phi_\tp}(\rho)>0.
\end{multline*}
According to Lemma~\ref{lemma:TPdetection} it is sufficient to have
$\dim\cW\leq\dim\cZ+1$.
Taking $\cY=\cW\oplus\mathbb{C}$ shows that it is sufficient to have
$\dim\cY\leq\dim\cZ+2$, and completes the proof.
\end{proof}

In regard to the type of channels that allow entangled states to give
improved discrimination, one has the following interesting corollary.
\begin{corollary}
  A state $\rho\in\density{\cX\ot\cZ}$ is entangled if and only if
  there exist entanglement-breaking channels
  $\Psi_0,\Psi_1\in\trans{\cX,\cY}$ such that
  \[
  \left\|\left(\Psi_0\otimes\I_{\lin{\cZ}}\right)[\rho] -
  \left(\Psi_1\otimes\I_{\lin{\cZ}}\right)[\rho]\right\|_{\rm tr}
  > \|\Psi_0-\Psi_1\|_{\rm tr}.
  \]
\end{corollary}
\begin{proof}
Generalizing the result of~\cite{Sacchi05b}, we observe that if an
entangled state $\rho\in\density{\cX\otimes\cZ}$ increases the
distinguishability of two channels $\Psi_0,\Psi_1\in\trans{\cX,\cY}$,
then it also increases the distinguishability of two entanglement
breaking channels of the form 
$\Xi_i^{p} = p \Psi_i + (1-p)\Omega$,
for $i=0,1$. 
Here $p\in(0,1]$ and $\Omega\in \trans{\cX,\cY}$ is the totally
depolarizing channel $\Omega[X]= (\Tr(X)/d_\cY) \I_{\cY}$. 

For sufficiently small $p>0$, the channels $\Xi_i^{p}$ are
entanglement breaking, as their Choi-Jamio{\l}kowski representations
are separable by the existence of a ball containing
only separable states around the maximally mixed
state~\cite{ZyczkowskiHSL98}.
It holds that
$\left\|\left((\Xi_0^{p}-\Xi_1^{p})\otimes\I_{\lin{\cZ}}\right)[\rho]\right\|_{\rm tr}
= p\|\left((\Psi_0-\Psi_1)\otimes\I_{\lin{\cZ}}\right)[\rho]\|_{\rm tr}$ 
and $\|\Xi_0^{p}-\Xi_1^{p}\|_{\mathrm{tr}}=p\|\Psi_0-\Psi_1\|_{\rm tr}$,
therefore the state $\rho$ enhances the
distinguishability of channels $\Xi_0^{p},\,\Xi_1^{p}$ for all choices of $p>0$.
\end{proof}

%--------------------------------------------------------------------------%
\section{Example}
%--------------------------------------------------------------------------%

The steps in the proof of Theorem~\ref{thm:main} are constructive. 
In particular, while the value of the enhancement in
distinguishability depends on the particular state, the channels that
are better distinguished by means of the states depend exclusively on
the positive map $\Phi_\tp$. We further remark that, for any entangled state, there exist tools to find a positive map that detects the state as entangled~\cite{DohertyPS04}. Unfortunately, it is not likely that this can efficiently be done~\cite{gurvits2003cdc,gharibian2008snh}.

The most well-known example of a positive linear map that detects
entanglement is transposition $T:L(\cX)\ni X \mapsto X^T\in L(\cX)$
with respect to some fixed basis of $\cX$~\cite{Peres96,HorodeckiHH96}.
For transposition one finds $c_{T}=2/(d_\cX+1)$, and channels
$\Psi_0,\Psi_1\in \trans{\cX,\cX\oplus \mathbb{C}}$,  
$\Psi_0:X\mapsto\frac{1}{d_\cX+1}((\Tr X) \I_\cX + X^T)$,
$\Psi_1:X\mapsto\frac{1}{d_\cX+1}((\Tr X) (\I_\cX + 2 \proj{0})- X^T)$,
with $\ket{0}$ a normalized vector orthogonal to $\cX$. Thus, for any
state $\rho\in D(\cX\otimes\cZ)$, we obtain
$\|\Psi_0\otimes\I_{\lin{\cZ}}[\rho]
-\Psi_1\otimes\I_{\lin{\cZ}}[\rho]\|_{\rm
  tr}-\|\Psi_0-\Psi_1\|_{\rm tr}=\frac{4}{d_\cX+1}N_T(\rho)$, with
$N_T(\rho)$ the standard negativity of
$\rho$~\cite{ZyczkowskiHSL98,VidalW02}.

%--------------------------------------------------------------------------%
\section{Conclusions}
%--------------------------------------------------------------------------%

We have proved that any entangled state is useful to distinguish some
pair of (entanglement-breaking) channels strictly better than what is possible by means of
a separable state in the minimum-error, single-shot scenario.
%This pair of channels may be taken to be arbitrarily noisy and able to
%destroy the entanglement of any arbitrary strongly entangled state. 
%This pair of channels may be taken to be able to
%destroy the entanglement of any entangled state. 
One may consider this result as a physically meaningful interpretation
of the characterization of entangled states by means of positive but
not completely positive linear maps~\cite{HorodeckiHH96}. We expect that our result
will stimulate further investigations on the role of entanglement in the discrimination of physical processes.

MP thanks C. E. Mora, P. Stelmachovic and M. Ziman for inspiring discussions.
We acknowledge support from NSERC, CIFAR, Quantumworks, and Ontario Centres of Excellence. 

%--------------------------------------------------------------------------%
%\bibliography{entdiscrim_v2}
%--------------------------------------------------------------------------%

\end{document}